\documentclass[11pt]{article}
\usepackage{amsmath,amsfonts,amsthm,amssymb,color}
\usepackage{framed,multicol,enumitem}
\usepackage{fancybox}
\usepackage{fullpage}

\newtheorem{theorem}{Theorem}[section]

\newtheorem{lemma}[theorem]{Lemma}

\newcommand{\suppress}[1]{}
\theoremstyle{definition}
\newtheorem{definition}[theorem]{Definition}

\newenvironment{fminipage}%
  {\begin{Sbox}\begin{minipage}}%
  {\end{minipage}\end{Sbox}\fbox{\TheSbox}}

\newenvironment{algbox}[0]{\vskip 0.2in
\noindent 
\begin{fminipage}{6.3in}
}{
\end{fminipage}
\vskip 0.2in
}

\def\defeq{\stackrel{\mathrm{def}}{=}}

\newcommand\cc{\boldsymbol{\mathit{c}}}

\newcommand\pp{\boldsymbol{\mathit{p}}}
\newcommand\qq{\boldsymbol{\mathit{q}}}

\newcommand\ww{\boldsymbol{\mathit{w}}}
\newcommand\yy{\boldsymbol{\mathit{y}}}
\newcommand\zz{\boldsymbol{\mathit{z}}}
\newcommand\xx{\boldsymbol{\mathit{x}}}
\newcommand\xxbar{\overline{\boldsymbol{\mathit{x}}}}

\newcommand\poly{\text{poly}}

\begin{document}
	
	\title{Concave Flow on Small Depth Directed Networks}
	
	\author{
	Tung Mai  \\
	Georgia Tech\\
	\texttt{tung.mai@cc.gatech.edu}
	\and
 	Richard Peng \\
	Georgia Tech\\
	\texttt{rpeng@cc.gatech.edu}
	\and 	
	Anup B. Rao \\
	Georgia Tech\\
	\texttt{arao89@gatech.edu}
	\and 
	Vijay V. Vazirani \\
	Georgia Tech\\
	\texttt{vazirani@cc.gatech.edu}
	}

	\maketitle

\begin{abstract}	
Small depth networks arise in a variety of network related applications,
often in the form of maximum flow and maximum weighted matching.
Recent works have generalized such methods to include costs arising from concave functions.
In this paper we give an algorithm that takes a depth $D$ network and strictly increasing concave
weight functions of flows on the edges and computes a  $(1 - \epsilon)$-approximation
to the maximum weight flow in time $mD \epsilon^{-1}$ times an overhead that
is logarithmic in the various numerical parameters related to the magnitudes
of gradients and capacities.

Our approach is based on extending the scaling algorithm for approximate
maximum weighted matchings by [Duan-Pettie JACM`14] to the setting of small
depth networks, and then generalizing it to concave functions.
In this more restricted setting of linear weights in the range
$[\ww_{\min}, \ww_{\max}]$, it produces a
$(1 - \epsilon)$-approximation in time $O(mD \epsilon^{-1} \log( \ww_{\max} /\ww_{\min}))$.
The algorithm combines a variety of tools and provides a unified approach
towards several problems involving small depth networks.
\end{abstract}

	\pagenumbering{gobble}
	
	\newpage
	
	\setcounter{page}{0}
	\pagenumbering{arabic}
	
	\section{Introduction}
\label{sec:intro}

Combinatorial problems have traditionally been solved using combinatorial
methods, e.g., see \cite{schrijver.book}.
Aside from providing many state-of-the-art results, these methods also have
the advantages of simplicity and interpretability of intermediate solutions.
Recently, algorithms that utilize continuous methods have led to performances
similar or better than these combinatorial methods~\cite{ChristianoKMST11,Madry13,AzO15}.
A significant advantage of continuous approaches is that they extend naturally
to general objective functions, whereas in combinatorial algorithms there
is much difference, for example, between even finding a maximum flow and
a minimum cost flow.

This paper is motivated by recent works that used accelerated gradient
descent to obtain $O(m\epsilon^{-1} \log{n} )$ time routines for computing
approximate solutions of positive linear programs~\cite{AzO15,WangRM15}.
Such linear programs include in particular (integral) maximum weighted matchings
on bipartite graphs, and fractional weighted matchings on general graphs.
These algorithms, as well as earlier works using continuous methods for the maximum
weighted matching problems~\cite{AhnG14}, represent a very different approach that
obtained bounds similar to the result by Duan and Pettie~\cite{DuanP14}; the latter
computed similar approximations using combinatorial methods.
We study this connection in the reverse direction to show that the
augmenting path based approaches can be extend to
1) small depth, acyclic networks, and
2) more general cost objectives.
Formally, given a DAG with source $s$ and sink $t$, along with
capacities $\cc_e$ and differentiable concave functions $f_e$
on each edge, we want to find a flow from $s$ to $t$ that
obeys capacity $\cc_e$ and maximizes the function
$\sum_{e} f_e(\xx_e)$.
Our main result to this direction is:
\begin{theorem}
\label{thm:main}
Given a depth $D$ acyclic network with positive capacities
and concave cost functions whose derivatives are between
$\ww_{\min}$ and $\ww_{\max}$ and queryable in $O(1)$ time,
\textsc{ConvexFlow} returns a $(1 - \epsilon)$-approximation to the maximum
weighted flow in $O \left( Dm \epsilon^{-1}   \log ( {\ww_{\max}} / {\ww_{\min}}) \log n \right)$ time. 
\end{theorem}
We note that the logarithmic dependency on $\ww_{\max} / \ww_{\min}$ is common
in scaling algorithms for flows: the case of linear costs corresponds to
$f_e(\xx_e) = \ww_e \xx_e$ for scalars $\ww_e$, and in turn
$\ww_{\min} = \min_e \ww_e$ and $\ww_{\max} = \max_e \ww_e$.
We will discuss the relation between this formulation and other flow
problems, as well this dependency on weights in Section~\ref{subsec:related}.

\begin{figure}[ht]
\vspace{-0.5cm}	
	\begin{algbox}
		\begin{enumerate}
			\item Initialize flow $\xx$, potentials $\pp$, step size $\delta$
			\item Repeat until $\delta$ is sufficiently small:
			\begin{enumerate}
				\item Repeat $O( 1 / \epsilon)$ steps:
				\begin{enumerate}
					\item Compute eligible capacities on the edges, forming the {\it eligible graph.}
					\item Send a maximal flow in the eligible graph.\label{step:blockingFlow}
					\item Recompute eligible graph, adjust potentials of vertices unreachable from $s$.
				\end{enumerate}
				\item Reduce $\delta$, and adjust the potentials.
			\end{enumerate}
		\end{enumerate}
	\end{algbox}
\vspace{-0.5cm}	
	\caption{Algorithmic Template for Our Routines}
	\label{fig:algoSketch}
\end{figure}

The main components of our algorithm are in Figure~\ref{fig:algoSketch}.
It involves the construction and flow routing on an \emph{eligiblility
graph}, which consists of all positive capacitiated edges calculated
from a set of eligibility rules.
Note that aside from the blocking-flow computation on this graph in
Step~\ref{step:blockingFlow}, all other steps are local
steps based on values on an edge, potentials of its two end points, and
an overall step size.
This algorithm in its fullest generalization is in Figure~\ref{alg:main}
in Section~\ref{sec:nonlinear}.

The paper is organized as follows. In Section~\ref{sec:simple}
we describe the simplest variant: max weight flow on unit-capacity, small depth
networks where the cost functions on edges are linear.
The main reason for presenting this simple setting is to show how we
trade running time with accuracy by appropriately relaxing eligibility
conditions for edges.
In Section~\ref{sec:scalingLinear}, we extend our algorithm to situations with
widely varying edge weights via a scaling scheme.
We then extend the result to concave functions in Section~\ref{sec:nonlinear},
leading to our main result as stated in Theorem~\ref{thm:main}.


\subsection{Our Techniques}
\label{subsec:techniques}

Our main technical contribution is a method for combining augmenting
paths algorithms with primal-dual schemes for flows with costs.
We combine blocking flows, which are essentially maximal sets
of augmenting paths, with the weight functions on edges via
eligibility rules that dictate which edge in the residual graph
can admit flow based on the current primal / dual solution.
Invariants that follow from these eligibility rules then
allow us to maintain approximate complementary
slackness between the primal and dual solutions.
In turn, these yield approximation guarantees.

Our eligibility rules are directly motivated by the scaling algorithms
for approximate maximum matching~\cite{DuanP14}, which is in turn
based on scaling algorithms for maximum weighted
matchings~\cite{GabowT89,GabowT91,Gabow17:arxiv}.
However, the fact that augmenting paths can contain edges from
different depth levels pose several challenges.
Foremost is the issue that a very heavily weighted edge may cause
flow to go through a lightly weighted edge.
As a result, we need to find a way to work around 
the key idea from~\cite{DuanP14} of
relating the weight of an edge to the first time a flow passes through it.
We do so by using paths as the basic unit by which we analyze our flows, and
analyzing the entire history of paths going through an edge (see Section~\ref{subsec:pathMap} for details).

We will develop our algorithms first in the case where all weights
on edges are linear, and then extend them to concave functions.
These extensions are obtained by viewing each edge with non-linear cost functions
as a collection of edges with differing weights and infinitesimal
capacities, and then simulating a continuous version of our linear
algorithm through blocking flows.

	\subsection{Related Works}
\label{subsec:related}

As our results combine techniques applicable to a wide
range of combinatorial problems, our results have similarities,
but also differ, from many previous works.
An incomplete list of them is in Figure~\ref{fig:related}.
In order to save space, we use $\tilde{O}(f(n))$ to hide
terms that are $polylog(f(n))$: this is a widely used convention
for the more numerical variants of these objectives, and here
mostly hides factors of $\log{n}$.
\begin{figure}
\begin{center}
\begin{tabular}{|c|c|c|c|c|}
\hline
reference & objective & graph & approximation & runtime \protect \footnotemark \\
\hline
\cite{GoldbergR98} & flow amount & any & $1 \pm \epsilon$ in value
	& $\tilde{O}(m^{3/2} \log(\epsilon^{-1}))$ \\
\cite{Orlin13} & flow amount & any & exact &  $\tilde{O}(mn)$\\
\cite{Cohen95} & flow amount & small depth & $1 \pm \epsilon$ in value
	&  $\tilde{O}(m D^2 \epsilon^{-3})$\\
\cite{Peng16} & flow amount & undirected & $1 \pm \epsilon$ in value
	&  $\tilde{O}(m \epsilon^{-3})$\\
{\bf our work} & flow amount & small depth & $1 \pm \epsilon$ in value
	& $\tilde{O}(m D \epsilon^{-1})$\\
\hline
\cite{LeeS14} &min linear cost & any & 
	$1 \pm \epsilon$ in value & $\tilde{O}(m^{3/2} \log(\epsilon^{-1}))$\\
\cite{Orlin93} &min linear cost & any & exact
	&  $\tilde{O}(m^2)$\\
\cite{GoldbergTarjan87} &min linear cost & any & exact
	&  $\tilde{O}(mn)$\\
	{\bf our work} &max linear weight & small depth & $1 \pm \epsilon$ in value
	& $\tilde{O}(m D  \epsilon^{-1} )$\\

\hline
\cite{GabowT89} & matching weight & any & $1 \pm \epsilon$ in value
	& $\tilde{O}(mn^{1/2} \log(\epsilon^{-1}))$ \\
\cite{DuanP14} & matching weight & any & $1 \pm \epsilon$ in value
	& ${O}(m \epsilon^{-1}\log(\epsilon^{-1}))$\\
\cite{AzO15} & matching weight & bipartite
	\protect \footnotemark
	& $1 \pm \epsilon$ in value
	& $\tilde{O}(m \epsilon^{-1} )$\\
\cite{AhnG14} & $b$-matching weight & any
	& $1 \pm \epsilon$ in value
	& $\tilde{O}(m \poly(\epsilon^{-1}))$\\
{\bf our work} & $b$-matching weight & bipartite & $1 \pm \epsilon$ in value
	& $\tilde{O}(m \epsilon^{-1} )$\\
\hline
\cite{Hochbaum07} & concave/convex & any & $1 \pm \epsilon$ pointwise 
	& $\tilde{O}(m^2 \log(\epsilon^{-1}))$\\
\cite{Vegh12} &min separable convex cost& any & exact & $\geq O(n^4)$\\
{\bf our work} & max concave weight& small depth & $1 \pm \epsilon$ in value
	& $\tilde{O}(m D \epsilon^{-1})$
	\protect \footnotemark 
	\\
\hline
\end{tabular}
\caption{Results Related to Our Algorithm 
}
\label{fig:related}
\end{center}
\end{figure}

While the guarnatees of many of the algorithms listed here
are stated for inputs with integer
capacities / costs~\cite{GabowT89,GoldbergR98,LeeS14}, their
running times are equivalent to producing $1 \pm \epsilon$-approximates
with $\poly(\log(\epsilon^{-1}))$ overhead in running time.

Our guarantees, which contain $\poly(\epsilon^{-1})$ in the
running time terms, closest resemble the approximation algorithms for
weighted matchings~\cite{DuanP14,AzO15,AhnG14}
and flows on small depth networks~\cite{Cohen95}.
These running times only outperform the routines with
$\poly(\log(\epsilon^{-1}))$ depdencies for mild ranges
of $\epsilon$.
In such cases, the solution can differ significantly point-wise
from the optimum: for example, deleting $\epsilon$-fraction
of a maximum cardinality matching still gives a solution that's
within $1 \pm \epsilon$ of the optimum.
As a result, the barriers towards producing exact solutions
in strongly polynomial time~\cite{Hochbaum07} are circumvented
by these assumptions, especially the bounds on the minimum
and maximum gradients.
These bounds are natural generalizations of the bounds on costs
in scaling algorithms for flows with costs~\cite{GabowT89}.
In some cases, they can also be removed via preprocessing steps
that find $\poly(n)$ crude approximations to the optimum,
and trim edge weights via those values~\cite{ChristianoKMST11,
DuanP14,Peng16}.
In particular, if all the capacities are integers, we can increase
the minimum weight, or gradient, in a way analogous to the matching
algorithm from~\cite{DuanP14}.

\addtocounter{footnote}{-2}
\footnotetext{	Many of the approximate algorithms assume integer capacities/weights,
	and have factors $\log C$ and/or $\log N$ in their running times,
	where $C$ and $N$ are the largest capacity and weight respectively.
	We omit those factors due to space constraints: they are discussed as dependencies
	on the accuracy $\epsilon$ in the text before Lemma~\ref{lem:padWMinGeneral}.}
\addtocounter{footnote}{1}
\footnotetext{This algorithm also can produce fractional matching solutions on general graphs.}
\addtocounter{footnote}{1}
\footnotetext{There is a hidden factor of logarithm of the ratio between max and min gradients.}

%
%
%

\begin{lemma}
\label{lem:padWMinGeneral}
If all capacities are integers and the functions $f_e$ are non-negative
and increasing, then we can assume
\[
\ww_{\min} \geq \frac{\epsilon \ww_{\max} }
{O\left(m^2 \cdot \sum_{e} \cc_e \right)},
\]
or $\ww_{\max} / \ww_{\min} \leq O ( \frac{m^2}{\epsilon} \cdot \sum_{e} \cc_e )$
without perturbing the objective by a factor of more than $1 \pm \epsilon$.
\end{lemma}

\begin{proof}
First, note that the assumption about integral capacities means that
there exists a feasible flow that routes at least $1 / m^2$ per edge,
and such a flow, $\xx'$, can be produced using a single depth-first
search in $O(m)$ time.

We will fix the gradients to create functions $\overline{f}$
whose gradients are all within $poly(m, \epsilon^{-1})$ of
\[
\widehat{w} \defeq \max_{e} f\left(\frac{\epsilon}{m^2}\right).
\]
First, note that $OPT \geq \widehat{w}$ since the functions are
non-negative.
So we can add $\frac{\epsilon \widehat{w}}{\sum_{e} \cc_e}$ to the
gradient of all the functions without changing the optimum by more
than $\epsilon \widehat{w} \leq \epsilon OPT$ additively.
Formally the functions:
\[
\overline{f}_e \left( x \right)
\defeq f_e\left(x \right) + x \cdot \frac{\epsilon \widehat{w}}{\sum_{e'} \cc_{e'}}.
\]
with associated maximum $\overline{OPT}$ satisfies
 $OPT \leq \overline{OPT} \leq (1 + \epsilon) OPT$.

To fix the maximum gradient, consider the truncated functions
\[
\widehat{f}_e\left(x \right)
\defeq
\begin{cases}
\overline{f}_e\left(x \right) & \qquad
\text{if $x \geq \frac{\epsilon}{m^2}$, and}\\
\frac{x m^2}{\epsilon} \cdot \overline{f}_e\left( \frac{\epsilon}{m^2} \right) & \qquad
\text{if $x \leq \frac{\epsilon}{m^2}$}
\end{cases}.
\]
It can be checked that $\widehat{f}_e(x) \leq \overline{f}_e(x)$ at all points,
and that the maximum gradient of $\widehat{f}_{e}(x)$ is bounded by
\[
\frac{x m^2}{\epsilon} \cdot \overline{f}_{e}\left( \frac{\epsilon}{m^2} \right)
\leq O\left(m^2 \epsilon^{-1} \widehat{w}\right).
\]
So it remains to show that $\widehat{OPT}$, the maximum with the functions
$\widehat{f}$ satisfies $\widehat{OPT} \geq (1 - \epsilon) \overline{OPT}$.
For this, consider the optimum solution to $\overline{OPT}$, $\overline{\xx}$,
the solution
\[
\left( 1 - \epsilon \right) \overline{\xx} + \epsilon \xx'
\]
is feasible because it's a linear combination of two feasible solutions,
has objective at least $1- \epsilon$ times the objective of $\overline{\xx}$
due to the concavity and monotonicity of $f$, and has flow at least
$\epsilon / m^2$ per edge.
The result then follows from observing that $\widehat{f}_e(x) = \overline{f}_e(x)$
for all $x \geq \epsilon / m^2$.
\end{proof}

\paragraph{Remark} By Lemma~\ref{lem:padWMinGeneral}, if all capacities are integer, 
the running time of our result can be written as 
$\tilde{O}(Dm \epsilon^{-1} \log C )$, 
where $C$ is the largest capacity.

Our assumption on the topology of the graph follows
the study of small depth networks from~\cite{Cohen95},
which used the depth of the network to bound the number
of (parallel) steps required by augmenting path finding routines.
From an optimization perspective, small depth is also natural
because it represents the number of gradient steps required
to pass information from $s$ to $t$. Gradient steps, in the absence of preconditioners,
	only pass information from a node to its neighbor, so diameter
	is kind of required for any purely gradient based methods
The acyclicity requirement is a consequence of our
adaptation of tools from approximating weighted
matchings~\cite{DuanP14}, and represent a major shortcoming
of our result compared to exact algorithms for minimum cost
flow~\cite{Orlin93} and convex flows~\cite{Hochbaum07,Vegh12}.
However, as we will discuss in Appendix~\ref{sec:applications},
there are many applications of flows where the formulations
naturally lead to small depth networks.

	\section{Simple Algorithm for Unit Capacities and Linear Weights}
\label{sec:simple}

In this section, we propose a simple $(1-\epsilon)$ approximation algorithm
for the case of $f_e(\cdot)$ being linear functions,
and all capacities are unit.
This assumption means we can write the cost functions as:
\[
f_e(\xx_e) = \ww_e \xx_e
\]
for some $\ww_e  \geq 0$.
Under these simplifications, the linear program (LP) that we are solving and its dual become:
\begin{framed}
\setlength{\columnsep}{20pt}
  \begin{multicols}{2}
	\begin{minipage}[0pt]{10pt}
	\begin{align*}
	\max \qquad & \sum_{e} \ww_e\xx_e\\
	\forall u, \qquad & \sum_{vu} \xx_{vu} = \sum_{uv} \xx_{uv}\\
	\forall e, \qquad & 0 \leq \xx_e \leq 1
	\end{align*}
	\end{minipage}

	\columnbreak

	\begin{minipage}[0pt]{10pt}
	\begin{align*}
	\min \qquad & \sum_{e} \yy_{e}\\
	\forall e, \qquad & \pp_{v} - \pp_{u} + \yy_{e} \geq \ww_{e} \\
	& \pp_s = \pp_t \\
	\forall e, \qquad & \yy_{e} \geq 0.
	\end{align*}	
	\end{minipage}
  \end{multicols}
\end{framed}



We add a dummy directed edge from $t$ to $s$ with infinite capacity to ensure flow conservation at each vertex. 
For every edge $e$ we mention in the paper, we implicitly assume
that $e = uv$ is an edge in the original network and $e \not = ts$, unless otherwise stated.
The dual inequality involving $\yy_e$ is equivalent to:
\[
\yy_e \geq \ww_{e} + \pp_{u} - \pp_{v},
\]
The RHS is critical for our algorithm and its analysis.
We will define it as the reduced weight of the edge, $\ww_{e}^{\pp}$:
\[
\ww_{e}^{\pp} = \ww_{e} + \pp_{u} - \pp_{v}.
\]
Note that the choice of $\ww_{e}^{\pp}$ uniquely determines
the optimum choice of $\yy_e$ through
$\yy_e \leftarrow \max \left( 0, \ww_{e}^{\pp} \right)$.

To find an optimum flow on these networks
it is sufficient  to obtain feasible pairs $\xx$ and $\pp, \yy$ that satisfy
the following complementary slackness conditions:
\begin{enumerate}
	\item For all $e$, if $\xx_e = 0$, then $\yy_e = 0$,
		which is equivalent to $\ww^{\pp}_{e} \leq 0$.
	\item For all $e$, if $\xx_e = 1$ then $\yy_{e} = \ww^{\pp}_{e}$,
		which is equivalent to $\ww^{\pp}_{e} \geq 0$. 
\end{enumerate}

The role of these values of $\ww_{e}^{\pp}$ can be further
illustrated by the following Lemma, which
is widely used in scaling algorithms.
\begin{lemma}
	\label{lem:reducedCost}
	For any flow $\xx$ for which the flow conservation is satisfied,
	and any potential $\pp$ such that $\pp_t = \pp_s$, we have:
	\[
		\sum_{e} \ww_e \xx_e =
			\sum_{e} \ww_{e}^{\pp} \xx_e.
	\]
\end{lemma}

\begin{proof}
	Expanding the summation gives
	\[
	\sum_{e} \ww_{e}^{\pp} \xx_e = \sum_{e} \ww_e \xx_e +
	\sum_{e = uv} \xx_e \pp_u - \sum_{e = uv} \xx_e \pp_v.
	\]
	Rearranging the last terms gives
	\[
	\sum_{u} \pp_{u} \left( \sum_{uv} \xx_{uv} - \sum_{vu} \xx_{vu} \right),
	\]
	which we know is $0$ for every $u \neq s, t$.
	We also have that the amount of flow leaving $s$
	is same as the amount of flow entering $t$, and $\pp_s = \pp_t$
	so the total contribution of these terms is $0$.
\end{proof}

 Note that if $\pp_s=0$ and $\ww_{e}^{\pp}$ is close to $0$ for each $e$ in a path from $s$ to some $u$,
the total reduced cost of the path is close to $\pp_u$.
This is the view taken by most primal-dual algorithms for approximate
flows: conditions are imposed on valid paths in order to make
$\pp_u$s emulate max weight residual paths to the source.
These routines then start with $\pp_t$ set to very large,
and maintaining a complementary primal/dual pair while
gradually reducing $\pp_t$.
Shortest path algorithms such as the network simplex method
and the Hungarian algorithm fall within this framework
~\cite[Chapter~12]{AhujaMO93}.
Such algorithms produce exact answers, but may take
$\Omega(n)$ steps to converge even in the unit weighted case
due to very long augmenting paths.

The central idea in our algorithms is to regularize these
complementary slackness conditions.
Specifically, we relax the $\xx_e = 0$ case from $\ww_{e}^{\pp} \leq 0$
to $\ww_{e}^{\pp} \leq \delta$, for some $\delta$ which we eventually
set to $\epsilon \ww_{\min}$. We may assume that all the weights are integer multiples of $\delta$ since the 
additive rounding error of each weight $\ww_e$ is at most $\epsilon \ww_e$.
We use a set of eligibility rules to define when we can send flows
along edges, or whether an arc is eligible.
The modified eligibility conditions are shown in Figure~\ref{fig:elgSimple}.

\begin{figure}[ht]
	\begin{algbox}
		\begin{enumerate}
			\item Forward edge: if $\xx_{e} = 0$, can send flow along $u \rightarrow v$
			if $\ww_{e}^{\pp} = \delta$.
			\item Backward edge: if $\xx_{e} = 1$, can send flow back from $v \rightarrow u$ if $\ww_{e}^{\pp} = 0$.
		\end{enumerate}
	\end{algbox}
	\caption{Eligiblity Rules for
		Simple Algorithm on Unit Capacity Networks
		with Linear Weights.}
	\label{fig:elgSimple}
\end{figure}

\begin{figure}[ht]
\begin{algbox}
	\begin{enumerate}[ label = (A\arabic*),ref=(A\arabic*)]
		\item \label{a1} If $\xx_{e} = 1$, then $\ww^{\pp}_{e} \geq 0$. 
		\item \label{a2} If $\xx_{e} = 0$, then $\ww^{\pp}_{e}  \leq \delta$.
	\end{enumerate}
\end{algbox}
\caption{Relaxed Invariant for
Simple Algorithm on Unit Capacity Networks
with Linear Weights.}
\label{fig:invSimple}
\end{figure}

These rules leads our key construct, the eligible graph:
\begin{definition}
\label{def:eligibleGraph}
Given a graph $G$ and a set of eligibility rules, the \emph{eligible
graph} consists of all eligible residual arcs with positive capacities, equaling
to the amount of flow that can be routed on them. A vertex $v$ is an \emph{eligible vertex} if there exists a path from $s$ to $v$ in the eligible graph.
\end{definition}

The main idea behind working on this graph is that this
approximate invariant can be maintained analogously to
exact complementary slackness.

This leads to an instantiation of the algorithmic template from
Figure~\ref{fig:algoSketch}.
Its pseudocode is given in Figure~\ref{alg:simple}.




\begin{figure}[ht]

\begin{algbox}
$\xx = \textsc{SimpleFlow}\left(G, s, t,
	\left\{\ww_e\right\}, \ww_{\min}, \ww_{\max}, \epsilon \right)$
	
		\textbf{Input:} Unit capacity directed acyclic graph $G$ from $s$ to $t$,\\
		\qquad Weights $\ww_e$ for each edge $e$, along with 
			 bounds $\ww_{\min} = \min_{e} \ww_{e}$ and $\ww_{\max} = \max_{e} \ww_e$.\\
		\qquad	 Error tolerance $\epsilon$.
		
		\textbf{Output:} $(1 - \epsilon)$ approximate maximum weight flow $\xx$.

\begin{enumerate}
	\item Compute levels $l_u$ of the vertices $u$ (length of a longest path from $s$ to $u$). \\
		Initialize $\xx_{e} = 0$, $\pp_u = \ww_{\max} \cdot l_u$, $\delta = \epsilon \ww_{\min}$.
	\item Repeat until $\pp_t = 0$:
	\begin{enumerate}
		\item \emph{Eligibility labeling:} Compute the eligible graph.
		\item \emph{Augmentation:} Send 1 unit of flow along a maximal collection of edge-disjoint paths in the eligible graph.
		\item \emph{Dual Adjustment:}  Recompute the eligible graph,
		for each vertex $u$ unreachable from $s$,
		set $\pp_{u} \leftarrow \pp_{u} - \delta$.
	\end{enumerate}
\end{enumerate}

\end{algbox}

\caption{Simple Algorithm Based on the Eligibility Conditions
From Figure~\ref{fig:elgSimple}}

\label{alg:simple}

\end{figure}

The key property given by the introduction of $\delta$ is after sending
$1$ unit of flow along a maximal collection of paths, $t$ is no longer reachable from $s$
in the eligible graph.
We can show this from the observation that all edges in a path
that we augment along become ineligible.
For an eligible forward edge $uv$, $\xx_{uv} = 0$ and $\ww_{uv}^{\pp} = \delta$. Therefore, after augmentation, $\xx_{uv} = 1$ and $vu$ is ineligible. A similar argument can be applied for an eligible backward edge $vu$ to show that after augmentation the residual forward edge $uv$ is also ineligible.

This fact plays an important role in our faster convergence rate:
\begin{lemma} The running time of the algorithm given in Figure~\ref{alg:simple} is $O(D \ww_{\max} m/\epsilon \ww_{\min} )$ where $D$ is the depth of the network.
\end{lemma}
\begin{proof} The algorithm starts with $\pp_t = D \ww_{max}$ and in each iteration, $\pp_t$ is decreased by $\delta$. Therefore, the number of iterations is $D \ww_{\max} / \delta$. Since in each iteration it takes $O(m)$ to find a maximal set of disjoint augmenting paths using depth-first search, the total running time is $O(D \ww_{\max} m/\epsilon \ww_{\min} )$.
\end{proof}

We also need to check that our eligibility rules and augmentation
steps indeed maintain the invariants throughout this algorithm.

\begin{lemma}
	\label{lem:simpleInvariants}
	Throughout the algorithm, the invariants from Figure~\ref{fig:invSimple} are maintained.
\end{lemma}
\begin{proof}
	
	It is easy to check that at the beginning of the algorithm, the invariants are satisfied.
	We will show by induction that dual adjustments and augmentations do not break them.
	Assume the two conditions hold before a dual adjustment step.
	Notice that a dual adjustment can only change reduced weights of edges that have exactly one eligible endpoint.
	For an edge $e = uv$, we consider 2 cases:
	\begin{itemize}
		\item $u$ is eligible, and $v$ is ineligible: a dual adjustment step decreases $\pp_v$, and thus increases $\ww^{\pp}_{e}$. Condition \ref{a1} still holds because the reduced weight  increases. Condition \ref{a2} still holds because reduced weight can only increase until it is equal to $\delta$ 
		and both $u$ and $v$ are eligible.
		\item $v$ is eligible, and $u$ is ineligible:  a dual adjustment step decreases $\pp_u$, and thus decreases $\ww^{\pp}_{e}$. Condition \ref{a1} still holds because the reduced weight can only decrease until it is equal to 0 
		and both $u$ and $v$ are eligible. Condition \ref{a2} still holds because the reduced weight decreases. 
	\end{itemize}
	Now assume that the two conditions hold before an augmentation step, and notice that we only send flow along eligible edges. For an edge $e$ with $\xx_{e} = 0$ before augmentation, the eligibility condition is $ \ww_{e}^{\pp} = \delta$, thus after augmentation $\xx_{e} = 1$ and condition \ref{a1} is satisfied. For an edge $e$ with  $\xx_{e} = 1$ before augmentation, the eligibility condition is $ \ww_{e}^{\pp} = 0$, thus after augmentation $\xx_{e} = 0$ and condition \ref{a2} is satisfied.
\end{proof}

It remains to bound the quality of the solution produced.
Once again, the invariants from Figure~\ref{fig:invSimple}
play crucial roles here.



 
\begin{lemma} \label{lem:relaxedComplementarity}
	If all the invariants and feasibility conditions hold 
	when the algorithm terminates with $\pp_s = \pp_t$,
	then the solution produced, $\xx$, is feasible and $(1 - \epsilon)$-optimal.
\end{lemma}

\begin{proof}
	The flow $\xx$ is feasible because in each augmentation step we only send one unit of flow along a path in the residual graph.
	
	Let $\xx^*$ be any other $st$ flow obeying the capacities.
	Our goal is to show that
	\[
	\sum_{e} \ww_e \xx_e \geq \left( 1 - \epsilon \right) \sum_{e} \ww_e \xx^*_e.
	\]
	
	To do so, we first invoke Lemma~\ref{lem:reducedCost}, which
	allows us consider the weights $\ww_{e}^{\pp}$ instead.
	We invoke the invariants based on the two cases:
	\begin{enumerate}
		\item If $\xx_e = 1$, then $\xx^*_e \leq \xx_e$
			since the edges are unit capacitated.
			The invariant also gives $\ww_e^{\pp} \geq 0$, which implies
				\[
				\xx_e \ww_e^{\pp} \geq \xx^*_e \ww_e^{\pp}.
				\]
		\item If $\xx_e = 0$, then $\ww_e^{\pp} \xx_e = 0$.
			The invariant gives $\ww_e^{\pp} \leq \delta$,
			which combined with $\xx^*_e \geq 0$ gives:
			\[
				\xx_e \ww_e^{\pp} = 0 \geq \xx^*_e \left( \ww_e^{\pp}
					- \delta \right).
			\]
			Furthermore, since $\delta \leq \epsilon \ww_e$,
			this also implies
			\[
				\xx_e \ww_e^{\pp} \geq \xx^*_e \ww_e^{\pp}
					- \epsilon \ww_e \xx^*_{e}.
			\]
	\end{enumerate}
	Summing the implications of these invariants across all edges gives
	\[
		\sum_{e} \xx_e \ww_e = \sum_{e} \xx_e \ww_e^{\pp}
		\geq \sum_{e} \xx^*_e \ww_e^{\pp} - \epsilon \sum_{e} \xx^*_e \ww_e
		= \left( 1 - \epsilon \right) \sum_{e} \xx^*_e \ww_e.
	\]
\end{proof}

	\section{Scaling Algorithm for Linear Weights}
\label{sec:scalingLinear}

In this section, we show how to improve the running time of the algorithm in the previous section through scaling.
The scaling algorithm makes more aggressive changes to vertex
potentials $\pp_u$ at the beginning, and moves to progressively smaller changes.
To be more precise, the algorithm starts with $\delta_{0} = \epsilon \ww_{max}$
and ends with $\delta_{T} = \epsilon \ww_{min}$,
and the value of $\delta$ is halved each scale.
Assume $\epsilon$ and $\ww_{\max} / \ww_{\min}$ are powers of 2, this immediately implies
\[
T = \log\left({\ww_{\max}} / {\ww_{\min}}\right).
\]

The vertex potentials are also initialized in a similar way, i.e $\pp_u = \ww_{\max} \cdot l_u$ for all $u$.
Scale $i > 0$ starts with
\[
\pp_t = {D\ww_{max}} / {2^{i}} + 2 D\delta_{i}
\]
and ends when
\[
\pp_t = D\ww_{max}/2^{i+1}. 
\]
The change in step sizes leads to more delicate invariants.
Instead of having a sharp cut-off at $\delta$ as in the
single-scale routine, we further utilize $\delta_{i}$ to
create a small separation between edges eligible in the forward
and backward directions.

Between scales, we make adjustments to all potentials
to maintain the invariants.
Pseudocode of this routine is given in Figure~\ref{alg:scalingLinear}.
\begin{figure}[ht]
	\begin{algbox}
		$\textsc{ScalingFlow}\left(G, s, t, 
			\left\{\ww_e, \cc_e\right\}, \ww_{\min}, \ww_{\max}, \epsilon \right)$
			
		\textbf{Input:} Directed acyclic graph $G$ from $s$ to $t$,\\
		\qquad weights $\ww_e$ for each edge $e$, along with 
		bounds $\ww_{\min} = \min_{e} \ww_{e}$ and $\ww_{\max} = \max_{e} \ww_e$,\\
		capacities $\cc_e$ for each edge $e$, \\
		\qquad	 error tolerance $\epsilon$.
			
		\textbf{Output:} $(1 - \epsilon)$ approximate maximum weight flow $\xx$.
			
		\begin{enumerate}
			\item Compute levels $l_u$ of the vertices $u$ (length of a longest path from $s$ to $u$). \\
			  	Initialize $\xx = 0, \pp_u = \ww_{\max} \cdot l_u$, $\delta_0 = \epsilon \ww_{max}.$
			
			\item For i from 1 to $T = \log \frac{\ww_{max}}{\ww_{min}}$ :
			\begin{enumerate}
				\item Repeat until $\pp_t = \frac{D \ww_{max}}{2^{i+1}}$ if $i<T$ or until $\pp_t = 0$ if $i = T$:	\label{ln:outerLoop}
				\begin{enumerate}
					\item \emph{Eligibility labeling:} Compute the eligible graph based on the eligibility rules:
						\begin{enumerate}
							\item $uv$ is eligible in the forward direction if $\delta_{i} \leq \ww_{uv}^{\pp}$,
							\item a backward edge $vu$ is eligible if $\ww_{uv}^{\pp} \leq 0$.
						\end{enumerate}
					\item \emph{Augmentation:}
					Send a blocking flow in the eligible graph.\label{ln:augmentation}
					\item \emph{Dual Adjustment:}  Recompute the eligible graph,
				for each vertex $u$ unreachable from $s$,
				set $\pp_{u} \leftarrow \pp_{u} - \delta_i$. \label{ln:dualAdjust}
				\end{enumerate}
				\item If $i<T$:
				\begin{enumerate}
					\item \emph{Dual Rescale:} $\pp_u \leftarrow \pp_u + \delta_{i} l_u$, \label{ln:dualAdjust} \label{ln:dualRescale}
					\item $\delta_{i+1} \leftarrow \frac{\delta_{i}}{2}.$
				\end{enumerate}
			\end{enumerate}
		\end{enumerate}
	\end{algbox}
	\caption{Scaling Algorithm Based on the Eligibility Condition From Figure \ref{fig:invElgScaling}.}
	\label{alg:scalingLinear}
\end{figure}


The rest of this section is a more sophisticated analysis
built upon Section~\ref{sec:simple} and the scaling algorithms
from~\cite{DuanP14}, leading to the following main claim:
\begin{theorem} \label{thm:factorLinear}
	For any $\epsilon < 1/10$, \textsc{ScalingFlow} returns $\xx$
	that is an $(1 - 8 \epsilon)$-approximation in
	$O \left( {Dm \epsilon^{-1}}   \log ( {\ww_{max}} / {\ww_{min}})  \log n \right)$ time.
\end{theorem}
		
A main issue in setting up the invariants, as well as the overall proof,
is that errors from the iterations accumulate.
Here we need an additional definition about the last time forward
flow as pushed on an edge.
\begin{definition}
	\label{def:prev}
	During the course of the algorithm, for each edge $e = uv$, define $scale(e)$
	to be the the largest scale $j$ such that $j$ is at most $i$, the current scale, and there
	is a forward flow from $u$ to $v$ at scale $j$. Define $scale^f(e)$ to be the final value of $scale(e)$
	when the algorithm terminates.
\end{definition}
Note that $scale(e)$ resets to $i$ if $uv$ is used in 
the blocking flow generated on Line~\ref{ln:augmentation} at scale $i$. For simplicity, we use $\delta_{e}$ to denote $\delta_{ scale({e})}$, $\delta^f_{e}$ 
to denote $\delta_{ scale^f({e})}$, and $l_e$ to denote $l_v - l_u$.
The key eligibility rules and invariants are in Figure~\ref{fig:invElgScaling}.


\begin{figure}[ht]
	
	\begin{algbox}
		\begin{enumerate}
			\item Forward edge: if $\xx_{e} < \cc_{e}$, $u \rightarrow v$ is eligible if $\ww_{e}^{\pp} \geq \delta_{i}$.
			\item Backward edge: if $\xx_{e} > 0$, $v \rightarrow u$ is eligible if $ \ww_{e}^{\pp} \leq 0$.
		\end{enumerate}
	\end{algbox}
	
	\begin{algbox}
		\begin{enumerate}[ label = (B\arabic*),ref=(B\arabic*)]
			\item \label{rrc1a} If $\xx_{e} > 0$, then $\ww^{\pp}_{e}
				\geq - 3 l_e (\delta_{e} - \delta_{i}) - \delta_i$. 
			\item \label{rrc2a} If $\xx_{e} < \cc_{e}$, then
				$\ww^{\pp}_{e} \leq 2 \delta_{i}$.
		\end{enumerate}
	\end{algbox}
	
	\caption{Relaxed Invariant and Corresponding Eligibility Rules for Scaling Algorithm on Non-unit  Capacity Networks with Linear Weights.}
	\label{fig:invElgScaling}
\end{figure}

The intricacies in Condition~\ref{rrc1a} is closely related to the
changes in the dual rescaling step from Line~\ref{ln:dualRescale}.
In particular, Condition~\ref{rrc2a}, $\ww_{e}^{\pp} < 2 \delta_{i}$,
may be violated if $\delta_{i}$ is halved.
The dual rescaling on Line~\ref{ln:dualRescale} is in place precisely
to fix this issue by lowering $\ww_{e}^{\pp}$.
It on the other hand causes further problems for lower bounds
on $\ww_{e}^{\pp}$, leading to modifications to Condition~\ref{rrc1a} as well.
We first verify that these modified invariants are preserved by
the dual rescaling steps.

\subsection{Dual Rescaling Steps}
\label{subsec:dualRescale}

\begin{lemma}
The dual rescaling step on Line~\ref{ln:dualRescale} of~Figure~\ref{fig:invElgScaling}
maintain invariants~\ref{rrc1a}~and~\ref{rrc2a}.
\end{lemma}

\begin{proof}
	
	
We start with Condition~\ref{rrc2a}.
If $\ww_{e}^{\pp} \leq 2\delta_{i}$, then if we pick $\pp'$
such that  $\ww_{e}^{\pp'} \leq \ww_{e}^{\pp} - \delta_{i}$, we get:
\[
\ww_{e}^{\pp'} \leq \delta_{i} = 2 \delta_{i + 1}.
\]
This is done by updating vertex potential
$\pp_u \leftarrow \pp_u + \delta_{i} l_u$ for each vertex $u$ at the end of scale $i$.

It remains to check that the condition
\[
\ww^{\pp}_{e} \geq - 3 l_e (\delta_{e} - \delta_{i}) - \delta_i
\]
does not break by decreasing $\pp_u - \pp_v$ by $l_e\delta_{i}$:
\begin{align*}
& - 3 l_e (\delta_{e} - \delta_{i}) - \delta_i - l_e\delta_{i} \\ 
& = - 3l_e \delta_{e} + 2 l_e \delta_{i} - \delta_i \\
& = - 3l_e \delta_{e} + 4 l_e \delta_{i+1} - 2\delta_{i+1} \\
& \geq - 3l_e ( \delta_{e} - \delta_{i+1}) - \delta_{i+1}.
\end{align*} 
where the last line follows because $l_e \geq 1$.
\end{proof}

\subsection{Blocking Flows}
\label{subsec:blockingFlows}

Our algorithm also incorporates more general capacities $\cc_e$.
The changes caused by this in the linear programming formulations are:
\begin{enumerate}
	\item In the primal, the capacity constraints become $0 \leq \xx_e \leq \cc_e$, \item the dual objective becomes $\sum_{e} \cc_e \yy_e$.
\end{enumerate}
This extension can be incorporated naturally:
instead of finding a maximal set of edge-disjoint augmenting paths,
we find a blocking flow from $s$ to $t$.
Blocking flows can also be viewed as maximal sets of paths when the edges
are broken down into multi-edges of infinitesimal capacities.
They can be found in $O(m \log{n})$ time using dynamic trees~\cite{SleatorT83},
and therefore offer similar levels of theoretical algorithmic efficiency.

In the unit capacity case, an augmentation along any path in the residual graph destroys all eligible edges in the path, and the asymmetry of the eligibility
condition means the new reverse edges are also ineligible.
However, when we move to general capacities, an augmentation may
destroy just a single edge on the path.
As a result, it's simpler to work with the following equivalent definition
of blocking flows.
\begin{definition}
	\label{sec:blockingFlow}
	A \emph{blocking flow} on a graph $G$ is a flow such that every
	path from $s$ to $t$ goes through a fully saturated edge.
\end{definition}

We first verify that this definition also maintains all the invariants
both after augmentation and dual adjustments.
These proofs closely mirror those of Lemma~\ref{lem:simpleInvariants}.

\begin{lemma}
	\label{lem:augmentOkScaling}
	An augmentation by a blocking flow in the eligible graph,
	as on Line~\ref{ln:augmentation}, maintains the invariant.
\end{lemma}

\begin{proof}
	The algorithm only sends flow through eligible edges. Consider an edge $uv$ on a blocking flow
	\begin{enumerate}
		\item If the augmentation sends flow from $u$ to $v$, then $uv$ is an eligible forward edge, and $\ww_{uv}^{\pp} \geq \delta_i$. After augmentation, condition \ref{rrc1a} must hold.
		\item  If the augmentation sends flow from $v$ to $u$, then $vu$ is an eligible backward edge, and $0 \geq \ww^{\pp}_{uv}$. After augmentation, condition \ref{rrc2a} must hold.
		
	\end{enumerate}
	
	Therefore, all edges on a blocking flow satisfy the set of reduced weight conditions after augmentation.
\end{proof}

\begin{lemma}
	\label{lem:dualAdjOkScaling}
	A dual adjustment step, as specified on Line~\ref{ln:dualAdjust},
	maintains the invariant.
\end{lemma}

\begin{proof}
	Notice that a dual adjustment can only change the effective reduce weights of edges that have exactly one eligible endpoint. For an edge $e = uv$ in the graph, we consider 2 cases:
	\begin{itemize}
		\item $u$ is eligible, and $v$ is ineligible: a dual adjustment step decreases $\pp_v$, and thus increases $\ww^{\pp}_{e}$. 
		\begin{itemize}
			\item if $\xx_{e} = \cc_{e}$ then condition \ref{rrc1a} still holds because $\ww^{\pp}_{e}$ increases.
			\item if $\xx_{e} <  \cc_{e}$ then we know that $uv$ is in the residual graph and ineligible. Since $uv$ is ineligible, $\ww^{\pp}_{e}$ must be less than $\delta_i$. Condition \ref{rrc2a} still holds because $\ww^{\pp}_{e}$ can only increase until it is at least $\delta_i$ and $e$ becomes eligible. 
		\end{itemize}
		\item $v$ is eligible, and $u$ is ineligible:  a dual adjustment step decreases $\pp_u$, and thus decreases $\ww^{\pp}_{e}$. 
		\begin{itemize}
			\item if $\xx_{e} > 0$ then we know that $vu$ is in the residual graph and ineligible. Since $vu$ is ineligible, $\ww^{\pp}_{e}$ must be greater than $0$. Condition \ref{rrc1a} still holds because $\ww^{\pp}_{e}$ can only decrease until it is at most $0$ and $vu$ becomes eligible. 
			\item if $\xx_{e} = 0$ then condition \ref{rrc2a} still holds because $\ww^{\pp}_{e}$ decreases.
		\end{itemize}
	\end{itemize}
\end{proof}

Note that at the beginning of the algorithm, $\xx_{e} = 0$ and $\ww_{e}^{\pp} = \ww_e - l_e \ww_{\max}\leq 0$ for each edge $e$. Therefore, the invariants \ref{rrc2a} holds.
Induction then gives that these invariants hold throughout the course of the  algorithm.


\subsection{Bounding Weights of Paths}
\label{subsec:pathMap}

We can now turn our attention to the approximation factor.
The main difficulty in this proof is due to cancellations of flows
pushed at  previous scales.
Our proof relies on performing a \emph{path decomposition} 
on each blocking flow in the history of the final solution $\xx$.
We will view the final flow, $\xx$, as a sum of flows from
the collection $\mathcal{P}$.
\[
\xx = \sum_{P \in \mathcal{P}} \xx_{P}.
\]
Note here that $\xx_{P}$ is a quantity that specifies
the amount of flow on that path.
We also define $F_P$ to be the set of edges going forward on $P$,
and $B_P$ to be the set of edges going backward on $P$.
We use this decomposition to map backward flow to
forward flow in an earlier scale.
We will also use $scale({P})$ to denote the scale
at which $P$ was pushed, and $\delta_P$ to denote $\delta_{scale(P)}$.

Note that we will likely push flows on an edge in forward and backward directions
many times throughout the course of the algorithm.
Each of these instances belong to a different path, and is taken
care of separately in the summation.
We first establish a lower bound on the weight of each path
using the eligibility criteria: we show that the weight of the
path routed at phase $i$ can be lower bounded by terms
involving $\delta_i$, the step size at that phase.

\begin{lemma} \label{lem:eligible}
\[ \sum_{e \in F_P} \ww_e - \sum_{e \in B_P} \ww_e
	\geq \frac{ D \delta_{P}}{2\epsilon }\]
\end{lemma}

\begin{proof}
	Let $\pp$ be the vertex potentials from when flow was pushed
	along $P$. Note that $\pp$ may not be the final potentials. By the definition of 
	$\ww_e^{\pp}$, we have
	\[
		 \sum_{e \in F_P} \ww_e - \sum_{e \in B_P} \ww_e
		=  \sum_{e \in F_P} \ww^{\pp}_e - \sum_{e \in B_P} \ww^{\pp}_e+ \left( \pp_t - \pp_s \right).
	\]
	
	For the edges on the path, the eligibility rules give
	\begin{itemize}
		\item  $\ww^{\pp}_e \geq \delta_P \geq 0$ if $e \in F_P$, and
		\item  $-\ww^{\pp}_e \geq 0$ if $e \in B_P$.
	\end{itemize}
	Furthermore, our outer loop condition on Line~\ref{ln:outerLoop}
	ensures that at scale $i$, $\pp_t $ is at least
	$D\ww_{max}/ 2^{i+1} = D \delta_{i} / 2 \epsilon$, while
	$\pp_s$ is always $0$.
	Therefore the total is at least $ D \delta_{i} / 2 \epsilon$.
%
\end{proof}

This allows us to map the final inaccuracies of weights
caused by the Condition~\ref{rrc2a} of the invariants
to $\epsilon$ of the total weights.
The following proof hinges on path decompositions

\begin{lemma} \label{lem:charge}
Let $\xx$ be the final answer.
Then we have
\begin{align*}  
\sum_{e} 3l_e\delta^f_{e} \xx_{e}
\leq 6 \epsilon \sum_{e} \ww_{e} \xx_{e}.
\end{align*}
\end{lemma}

\begin{proof} 
Let $\xx_P$ be the amount of flow sent through $P.$ The definition of $\xx$ and summation of flows gives
\[
\xx_{e} = \sum_{P:e \in F_P}  \xx_{P} -  \sum_{P:e \in B_P} \xx_{P} .
\]
reversing the ordering of summation, and aggregating per edge gives:
\[
\sum_{e} 3l_e\delta^f_{e} \xx_{e} = \sum_{e} 3 l_e\delta^f_{e}
	\left(\sum_{P: e \in F_P}  \xx_{P}
		-  \sum_{P: e \in B_P} \xx_{P}\right).
\]

For each edge $e$, 
\begin{align}
 \delta^f_{e} \left(\sum_{P: e \in F_P}  \xx_{P} -  \sum_{P: e \in B_P} \xx_{P}\right) 
 = \sum_{P: e \in F_P} \delta^f_{e} \xx_{P}
	 -  \sum_{P: e \in B_P} \delta^f_{e} \xx_{P}  \label{eq}
\end{align}
Consider the set of paths $P_1, \ldots P_k$ sending flow on edge $e$ (in both directions) in the order of index $i$ for $i = 1 \ldots k$. Whenever a flow is sent from $v$ to $u$ (a backward direction), there must be some existing flow from $u$ to $v$. By breaking a path into multiple paths carrying smaller amounts of flow, we may assume that for each $P_i$ such that $e \in B_{P_i}$, there exists $P_j$ such that $j < i$, $e \in F_{P_j}$ and $\xx_{P_i} = \xx_{P_j}$. We define a function that maps $i$ to $j$. Since $j<i$, we must have $\delta_{P_j} \geq \delta_{P_i}$. It follows that 
\[ \delta^f_{e} \xx_{P_j} - \delta^f_{e} \xx_{P_i} = 0 \leq \delta_{P_j} \xx_{P_j} - \delta_{P_i} \xx_{P_i}  \]
We apply the following modifications to the RHS of equation (\ref{eq}). For each pair $(i,j)$ such that $i$ is mapped to $j$, we increase the multiplication factor of $\xx_{P_i}$ from $\delta^f_{e}$ to $\delta_{P_i}$ and the multiplication factor of $\xx_{P_j}$ from $\delta^f_{e}$ to $\delta_{P_j}$. Finally, for the remaining forward flows that are not canceled, we increase $\delta^f_{e}$ to the value $\delta_P$ of the path containing them. Such modifications can only increase the value of the RHS of (\ref{eq}). Therefore, for each edge $e$,

\[\sum_{P: e \in F_P} \delta^f_{e} \xx_{P} -  \sum_{P: e \in B_P} \delta^f_{e} \xx_{P} \leq \sum_{P: e \in F_P} \delta_{P} \xx_{P} -  \sum_{P: e \in B_P} \delta_{P} \xx_{P}. \]

Plugging back to the summation, we have 

\[\sum_{e} 3l_e\delta^f_{e} \xx_{e} \leq \sum_{e} 3 l_e
	\left(  \sum_{P: e \in F_P} \delta_{P} \xx_{P} -  \sum_{P: e \in B_P} \delta_{P} \xx_{P} \right).\]
Changing the order of summation and notice that for each path $P$
\[  \sum_{e \in F_P} l_e\delta_{P} -  \sum_{e \in B_P} l_e \delta_{P} = D\delta_P,\]
which by Lemma~\ref{lem:eligible} is bounded by 
$2 \epsilon \left( \sum_{e \in F_P} \ww_{e} -  \sum_{e \in B_P} \ww_e \right) $.
Summing over all paths $P \in \mathcal{P}$ then gives the result.
\end{proof}

\subsection{Putting Things Together}
\label{subsec:finish}

It remains to put these together through the use
of reduced costs in a way analogous to
Lemma~\ref{lem:relaxedComplementarity}.

\begin{proof}(of Theorem~\ref{thm:factorLinear})
	The number of iterations in each scale is ${D}/ {2\epsilon} + 2D = O(D / {\epsilon})$. The number of scales is $\log ( {\ww_{max}} / {\ww_{min}} )$. Each iteration takes $O(m \log n)$ to find a blocking flow. Therefore, the total running time is $O \left( ({Dm \log n} / {\epsilon} ) \log ( {\ww_{max}} / {\ww_{min}}) \right)$.
	
	For the approximation guarantees, consider any optimal flow $\xx^*$.
	Let $\pp$ denote the final potentials.
	Once again, Lemma~\ref{lem:reducedCost} gives
	$\sum_{e} \ww_{e}^{\pp} \xx_e = \sum_{e} \ww_e \xx_e$
	and $\sum_{e} \ww_{e}^{\pp} \xx^*_e = \sum_{e} \ww_e \xx^*_e$.
	
	We will incorporate the invariants from Figure~\ref{fig:invElgScaling}
	through two cases:
		\begin{enumerate}
			\item If $\xx_e < \xx^*_e$, then $\xx_e < \cc_e$, and
			the invariant gives $\ww_e^{\pp} \leq 2 \delta_{T}$, which implies
			\[
			\ww_e^{\pp}\xx_e  \geq \left( \ww_e^{\pp} - 2 \delta_{T} \right) \xx_e  \geq  \left( \ww_e^{\pp} - 2 \delta_{T} \right) \xx^*_e.
			\]
			Since $\delta_{T} = \epsilon \ww_{\min} \leq \ww_e$, this implies
			\[
			\xx_e \ww_e^{\pp} \geq \xx^*_e \ww_e^{\pp} - 2\epsilon \ww_{e} \xx^*_e .
			\]
			\item If $\xx_e > \xx^*_e$, then $\xx_e > 0$, and the invariant gives
			\[
				\ww_e^{\pp} \geq -3 l_e \delta_e^f
			\]
			which combined with $\xx^*_e \leq \xx_e$ gives:
			\[
			\xx_e \left( \ww_e^{\pp} + 3 l_e \delta_e^f \right) \geq \xx^*_e \left( \ww_e^{\pp} + 3 l_e \delta_e^f \right) \geq \xx^*_e \ww_e^{\pp}
			\]
		\end{enumerate}
		Combining these gives
		\begin{align*}
			\sum_{e} \xx_e \left( \ww_e + 3 l_e\delta_e^f  \right) & = \sum_{e} \xx_e \left( \ww_e^{\pp} + 3 l_e\delta_e^f  	\right)  \\
		 & \geq \sum_e \xx^*_e \ww_e^{\pp} - 2 \epsilon \sum_e \ww_e  \xx^*_e	 \\
		 & = (1 - 2 \epsilon) \sum_e \ww_e  \xx^*_e
		\end{align*}
		and the result then follows from Lemma~\ref{lem:charge}.
		
\end{proof}

	\section{Non-Linear Networks}
\label{sec:nonlinear}

We now move to the non-linear setting with
arbitrary capacities.
The primal objective becomes:
\begin{align*}
\max & \sum_{e} f_e(\xx_e)
\end{align*}
Here we will let $\ww_e(\xx_e)$ denote the weight of $e$
when $\xx_e$ units are on it, i.e.,
\[
\ww_{e}(\xx_e) = \frac{d}{d \xx_e} f_{e}(\xx_{e}).
\]
Since $f_{e}(\cdot)$ is concave, we know that
$\ww_{e}(\xx_e)$ is monotonically decreasing.
Pseudocode of our algorithm is in Figure~\ref{alg:main},
and its  eligibility rules are in Figure \ref{fig:eligScalingNonlinear}. 
Note that the eligibility rules in the pseudocode is a shorter version of the one in Figure \ref{fig:eligScalingNonlinear}.

\begin{figure}[ht]
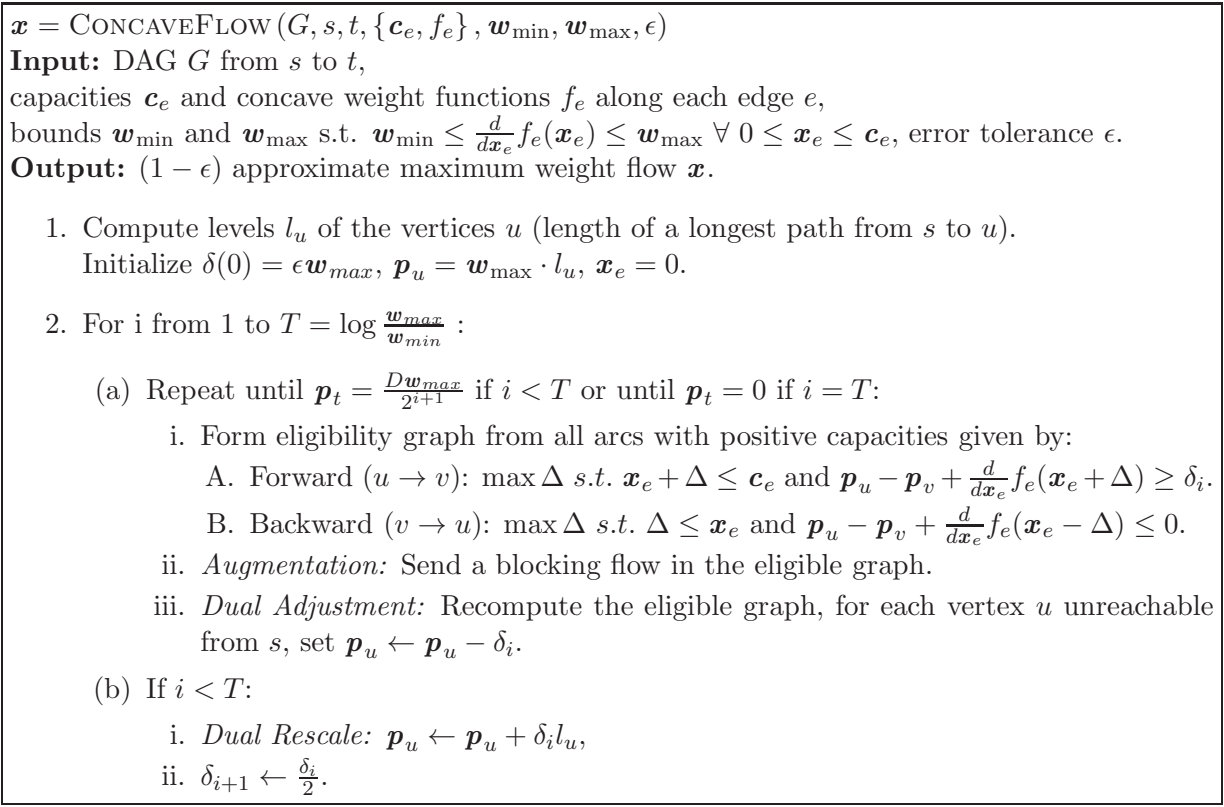

	
	\begin{algbox}
		$\xx = \textsc{ConcaveFlow}\left(G, s, t, \left\{ \cc_e, f_e \right\}, \ww_{\min}, \ww_{\max}, \epsilon \right)$
	
		\textbf{Input:} DAG $G$ from $s$ to $t$,\\
		\qquad capacities $\cc_e$ and concave weight functions $f_e$ along
		each edge $e$,\\
		\qquad bounds $\ww_{\min}$ and $\ww_{\max}$ s.t.
		$\ww_{\min} \leq \frac{d}{d \xx_e}f_e(\xx_e) \leq \ww_{\max}~
			\forall~0 \leq \xx_e \leq \cc_e$,
			 error tolerance $\epsilon$.
		
		\textbf{Output:} $(1 - \epsilon)$ approximate maximum weight flow $\xx$.
		\begin{enumerate}
			\item Compute levels $l_u$ of the vertices $u$ (length of a longest path from $s$ to $u$). \\
				Initialize $\delta(0) = \epsilon \ww_{max}$,  $\pp_u = \ww_{\max} \cdot l_u$, $\xx_e = 0$.
			
			\item For i from 1 to $T = \log \frac{\ww_{max}}{\ww_{min}}$ :
			\begin{enumerate}
				\item Repeat until $\pp_t = \frac{D \ww_{max}}{2^{i+1}}$ if $i<T$ or until $\pp_t = 0$ if $i = T$:	\label{ln:outerLoop}
				\begin{enumerate}
					\item Form eligibility graph from all arcs with positive capacities given by:
					\begin{enumerate}
					\item Forward ($u \rightarrow v$): $\max \Delta~s.t.~\xx_e + \Delta \leq \cc_e
					~\text{and}~
					\pp_u - \pp_v + \frac{d}{d \xx_e} f_e(\xx_{e} + \Delta) \geq \delta_{i}$.\label{item:forward}
					\item Backward ($v \rightarrow u$): $\max \Delta~s.t.~\Delta \leq \xx_e
					~\text{and}~
					\pp_u - \pp_v + \frac{d}{d \xx_e} f_e(\xx_{e} - \Delta) \leq 0$.\label{item:backward}
					
					\end{enumerate}
					\item \emph{Augmentation:}
					Send a blocking flow in the eligible graph.\label{ln:augmentation}
					\item \emph{Dual Adjustment:}  Recompute the eligible graph,
				for each vertex $u$ unreachable from $s$,
				set $\pp_{u} \leftarrow \pp_{u} - \delta_i$. \label{ln:dualAdjust}
				\end{enumerate}
				\item If $i<T$:
				\begin{enumerate}
					\item \emph{Dual Rescale:} $\pp_u \leftarrow \pp_u + \delta_{i} l_u$, \label{ln:dualAdjust} \label{ln:dualRescale}
					\item $\delta_{i+1} \leftarrow \frac{\delta_{i}}{2}.$
				\end{enumerate}
			\end{enumerate}
		\end{enumerate}
	\end{algbox}
	
	\caption{Algorithm for Concave Networks}
	
	\label{alg:main}
	
\end{figure}

We remark that the cost of finding $L_e$ and $U_e$,
the maximum amount the flow can increase/decrease before
the gradient of $f_e(\xx_e)$ changes by $\Delta$,
depends on the structures of $f_e(\cdot)$.
In case of simple functions such as quadratics, they can
be calculated in closed form.
For more complex functions, finding these points can be viewed
as analogs of line searches~\cite{BoydV04:book}.

\begin{figure}[ht]

\begin{algbox}
\begin{enumerate}
\item Forward edge: if $\xx_{e} < \cc_{e}$,
$uv$ is eligible if $\delta_{i} < \ww_{e}^{\pp}(\xx_{e})$, and can has forward
capacity
\[
U_{e}(\xx_e, i) :=
 \min\left\{ \cc_{e} - \xx_{e},  \Delta  \text{ such that }\, \ww^{\pp}_e(\xx_{e} + \Delta) = \delta_{i} \right\}.
\]
\item Backward edge: if $\xx_{e} > 0$, $vu$ is eligible if $0 > \ww_{e}^{\pp}(\xx_{e})$, with capacity
\[
  L_e(\xx_e,i) :=
 \min \left \{ \xx_{e}, \Delta \text{ such that } \, \ww^{\pp}_e(\xx_{e} - \Delta) = 0 \right\}.
 \]
\end{enumerate}

\end{algbox}

\caption{Eligiblity Rules for Scaling Algorithm for Concave Weight Functions.}
\label{fig:eligScalingNonlinear}
\end{figure}

Conceptually, we divide each edge $uv$ into multiple imaginary parallel edges between $u$ and $v$, each with a linear weight function.
This can be viewed as approximating $\ww_e(\xx_e)$ by a decreasing step function. To be precise, for each edge $e$, we break it into $k$ parallel edges $e_1,\ldots , e_k$ such that:
\begin{enumerate}
	\item $\ww_{e_1} = \lfloor \ww_e(0) / \delta_{T} \rfloor \delta_{T}$, $ \ww_{e_k} = \lfloor \ww_e(\cc_e) / \delta_{T} \rfloor \delta_{T}$ and $\ww_{e_{i}} = \ww_{e_{i-1}} - \delta_{T}$ for all $1<i<k$.
	\item $\cc_{e_i} = \xx_{i} - \xx_{i-1}$ where $\xx_i$s are such that  $\xx_{0} = 0$,  $\xx_{k} = \cc_e$ and $\ww_e(\xx_i) = \ww_{e_i}$ for all $0<i<k$.
\end{enumerate}
Let $G = (V,E)$ be the original network, and $\overline{G}= (V,\overline{E})$ be the network with multiedges. We say that a flow $\overline{\xx}$ in $\overline{G}$ is \emph{well-ordered} if the following condition holds for every $e$:
\[ \overline{\xx}_{e_i} > 0 \Rightarrow \overline{\xx}_{e_j} = \cc_{e_j} \quad \forall j < i.\]
A flow $\overline{\xx}$ in $\overline{G}$ is \emph{equivalent} with a flow $\xx$ in $G$ if:
\[ \sum_j \xxbar_{e_j} = \xx_e. \]
Notice that the well-ordered property
ensures that there is a unique reverse mapping from $\xx_e$ to 
$\xx_{e_1} \ldots \xx_{e_k}$ as well.

\begin{lemma} \label{lem:mapping}
Let $(\xx,\pp)$ be a flow-potential pair at some point in the course of \textsc{ConcaveFlow} on input $G$. \textsc{ScalingFlow} on input $\overline{G}$ can maintain a flow-potential pair $(\xxbar, \pp)$ such that:
\begin{enumerate}
	\item $\xxbar$ is well-ordered.
	\item $\xxbar$ and $\xx$ are equivalent. 
\end{enumerate}
\end{lemma}
\begin{proof}  
It suffices to prove that the above claim holds in an Eligibility Labeling and an Augmentation step. 

Consider an edge $uv$ that is eligible in forward direction in \textsc{ConcaveFlow} algorithm. Since $\xxbar$ is well-ordered and equivalent to $\xx$, there are some eligible multiedges between $u$ and $v$ in \textsc{ScalingFlow}. Moreover, the first condition in Figure \ref{fig:eligScalingNonlinear} guarantees that the amount of flow allowed on $uv$ is equal to sum of capacities of all eligible parallel edges from $u$ to $v$ in the multiedges setting. This follows because 
\begin{enumerate}
	\item The value of $\ww_e(\xx_{e} + \Delta)$ such that  $\ww^{\pp}_e(\xx_{e} + \Delta) = \delta_{i}$ is a multiple of $\delta_i$, since $\pp_u$ and $\pp_v$ are also multiples of $\delta_i$ assuming $\epsilon$ is a power of 2.
	\item A weight of each edge in $\overline{G}$ is a multiple of $\delta_{T}$ by construction.
\end{enumerate}

A similar argument can be applied to the backward eligibility case. Specifically, the amount of flow that can be pushed back on $uv$ in \textsc{ConcaveFlow} is equal to sum of capacities of all eligible multiedges from $v$ to $u$ in \textsc{ScalingFlow}.

It follows that a blocking flow in the eligible graph in \textsc{ConcaveFlow} is also a blocking flow in the eligible graph in \textsc{ScalingFlow}. Given a blocking flow $\zz$ in the single-edge setting, we can distribute $\zz_{uv}$ among eligible multiedges between $u$ and $v$ such that the resulting flow $\xxbar$ in \textsc{ScalingFlow} is still well-order, and still equivalent to the resulting flow $\xx$ in \textsc{ConcaveFlow}.
\end{proof}

This mapping then allows us to bound the result of the continuous process,
leading to our main result.

\begin{proof}(Of Theorem~\ref{thm:main})

Let $\xx^*$ be the optimal solution in the concave flow problem in $G$, and $\xx$ be the final solution returned by $\textsc{ConcaveFlow}$. Let $\xxbar^*$ be a solution in $\overline{G}$ such that $\xxbar^*$ is well-ordered and equivalent to $\xx^*$. By Lemma~\ref{lem:mapping}, \textsc{ScalingFlow} on input $\overline{G}$ can produce a solution $\xxbar$ in such that $\xxbar$ is well-ordered and equivalent to $\xx$. At the end of the algorithm, the invariant in Section \ref{sec:scalingLinear} must hold for the solution $\xxbar$. By Theorem \ref{thm:factorLinear},
\[ \sum_{e \in \overline{G}} \ww_e \xxbar_e \geq (1 - 8 \epsilon) \sum_{e \in \overline{G}} \ww_e \xxbar^*_e.\]
Since $\xx$ is equivalent to $\xxbar$ and the step function defined by $\left \{ \ww_{e_i}, \cc_{e_i} \right \}$ is always below the function $\ww_e$,
\[ f_e(\xx_e) = \int_{0}^{\xx_e} \ww_e(t) dt \geq \sum_{e \in \overline{G}} \ww_e \xxbar_e.\]
Notice that the step function defined by $\left \{ \ww_{e_i} + \delta_{T}, \cc_{e_i} \right \}$, however, is always above the function $\ww_e$. Since $\xx^*$ is equivalent to $\xxbar^*$,
\[ \sum_{e \in \overline{G}} \left( \ww_e + \delta_{T} \right)\xxbar^*_e \geq \int_{0}^{\xx^*_e} \ww_e(t) dt = f_e (\xx^*_e).\]
Since $\delta_{T} = \epsilon \ww_{\min} \leq \epsilon \ww_e$ for all $e \in \overline{G}$, we have $\ww_e + \delta_{T} \leq  \ww_e(1 + \epsilon)$. It follows that
\[ (1+ \epsilon)\sum_{e \in \overline{G}} \ww_e \xxbar^*_e \geq f_e(\xx_e^*).\]
 Putting everything together, 
\begin{align*} f_e(\xx_e) & \geq \sum_{e \in \overline{G}} \ww_e \xxbar_e \geq (1 - 8 \epsilon) \sum_{e \in \overline{G}} \ww_e \xxbar^*_e \\
 & \geq \frac{1 - 8 \epsilon}{1 + \epsilon}f_e (\xx^*_e) \geq (1 - 9\epsilon)f_e (\xx^*_e).
\end{align*}
Finally, by applying a similar argument as in Theorem~\ref{thm:factorLinear}, it is easy to see that the running time of \textsc{ConcaveFlow} is $O \left( {Dm } \epsilon^{-1} \log ( {\ww_{max}} / {\ww_{min}}) \log n \right)$.
\end{proof}

\section{Discussion}

We extended the Duan-Pettie framework for approximating weighted
matchings~\cite{DuanP14} to flows on small depth networks~\cite{Cohen95}.
Our approach can be viewed as viewing the length of augmenting paths as a regularizing
parameter on the weights of augmenting paths.
For general concave weight functions, the routine can also be viewed
as a dynamical system that adjusts flows based on the gradient of the
current weight functions~\cite{BonifaciMV12,BecchettiBDKM13,StraszakV16}.
Also, we believe our approach can be parallelized,
and extends to non-bipartite graphs: here the challenge is mostly
with finding an approximate extension of blocking flow routines.

Another related application with small depth networks is finding high density
subgraphs in networks~\cite{GionisT15,MitzenmacherPPTX15,BhattacharyaHNT15}.
These problems can be solved using parameteric maximum flows
or binary searching over maximum flows~\cite{Goldberg84}.
In both of these cases, the networks have constant depth,
but the reductions also change the notions of approximation.
Weights represents an additional tool in such routines,
and we plan to further investigate its applications.

	\bibliographystyle{alpha}
	\bibliography{../ref}
	
	\appendix
	\section{Applications}
\label{sec:applications}

We now outline some concrete applications of our formulation on
small depth networks: many of them are discussed in more details in the
survey by Ahuja et al.~\cite{AhujaMO95}.

\subsection{Assignment Problems}

Our techniques are directly motivated by those developed
for approximate max weight matchings~\cite{DuanP14}.
This problem, when phrased as the assignment problem,
has applications in matching objects across snapshots
as well as direct uses in assignment indicated by its name.
Names of some of these applications, as enumerated in
Application 16-18 of ~\cite{AhujaMO95}, include
locating objects in space, matching moving objects,
and rewiring typewriters.

The problem of finding matchings in large scale graphs has received
much attention in experimental algorithm design~\cite{LattanziMSV11,ManshadiAGKMS13,KumarMVV15}.
These instances often involve capacities on both sides,
as well as edges in between.
They are related to $b$-matchings, and have been studied
as generalized matchings.
Previous works that incorporated concave costs include:

\begin{itemize}

\item ~\cite{VVJ10} studied Adwords matching using
concave objective functions, and showed several stability
and robustness properties of these formulations.

\item \cite{TangTLTGL12} utilized concave networks
to perform assignments of experts to queries with
additional constraints on groups of experts as well as queries.
This leads to a depth $4$ network with $3$ layers corresponding
to node groups respectively.

\item
\cite{devanur} considered a generalization of the well-known
Adwords problem, which asks for maximizing the sum of budgeted linear functions of agents.
In the generalization, the utility of each agents is a concave function (of her choice) of her total allocation.
The objective is to maximize the sum of budgeted utilities, and 
an optimal online algorithm is given for this generalization.

\end{itemize}
Another instance of concave cost matching is matrix balancing
(Application 28 in~\cite{AhujaMO95}):
given a fractional matrix, 
we would like to find a matrix with a specified row /column sum
while minimizing distance under some norm.
Here a bipartite network can be constructed by turning the
rows into left vertices, columns into right vertices.
The concave functions can then go onto the edges as costs.

Our algorithm is able to handle such generalizations,
and lead to bounds similar to those from recent improvements
to positive packing linear programs~\cite{AzO15}.

Generalizations of matching and assignment lead naturally
to small depth networks. 
One such generalization is 3-Dimensional matching
which is NP-hard~\cite{Karp72}.
Here the generalization of edges is $3$-tuples of the
form $(x, y, z)$, with one element each from
sets $X$, $Y$, and $Z$, and 3-dimensional matching
asks for the maximum weighted set of tuples that
uses each element in $X$, $Y$, and $Z$ at most once.

A restriction of this problem on the other hand is
in $P$: the set of acceptable tuples are given
by pairs of edges between $XY$ and $YZ$.
Specifically, a tuple $(x, y, z)$ can be used iff
$xy \in E_{XY}$ and $yz \in E_{YZ}$, and its
cost is given by the sum of the two edges.
This version can be solved by maximum cost
flow on a depth $5$ network:
\begin{itemize}
\item For each element $x, y$ or $z$, create an in and out vertex,
connected by an edge with unit capacity.
\item Connect $X_{out}$ to $Y_{in}$ and $Y_{out}$
to $Z_{in}$ with edges corresponding to $E_{XY}$ and $E_{YZ}$.
\item Connect $s$ to $X_{in}$ and $Z_{out}$ to $t$.
\end{itemize}

\subsection{Mincost Flow}

Minimum cost flow asks to minimize the cost instead
of maximizing it.
This is perhaps the most well-known formulation of our problem.
However, it doesn't tolerate approximations very well:
if we stay with non-negative costs, we need to specify
the amount of flow.
If this is done exactly, it becomes flow feasiblity
on a small depth network.
There are known reductions of directed maximum flows
to this instance (e.g.~\cite{Madry13}, Section 3).

One way to limit the total amount of flow is to associate
a reward with each unit sent from $s$ to $t$. Specifically, let $\qq_e > 0$ be the cost of one unit of flow on $e$, 
and $Q > 0$ be the reward of one unit sent from $s$ to $t$, the objective function is
\[ \min \sum_e \qq_e \xx_e - Q f \]
where $f$ is the total amount of flow from $s$ to $t$.
In this case, negating the cost brings us back
to the maximum weight flow problem.
The approximation factor on the other hand becomes additive:
\begin{lemma} \label{lem:mincost}
Let $(\xx^*,f^*)$ be the optimal solution to the above mincost flow problem, The algorithm \textsc{ScalingFlow} in Section \ref{sec:scalingLinear} gives a solution $(\xx,f)$ such that:
\[\sum_e \qq_e \xx_e - Q f \leq (1+\epsilon) \left(\sum_e \qq_e \xx^*_e - Q f^* \right) + 2 \epsilon Qf^*\]
\end{lemma}
A proof of Lemma \ref{lem:mincost} can be found in Appendix~\ref{sec:negative}. Applications of minimum cost flows that lead to small
depth networks include:
\begin{itemize}
	\item Job scheduling, or optimal loading of a hopping airplane (application 13 from~\cite{AhujaMO95}): 
	When phrased as scheduling, this can be viewed as completing a maximum set of jobs each taking place between days $[s_i, t_i]$ and
giving gain $c_i$ so that at no point more than $u_i$ jobs are being done simultaneously on day $i$.

	This can be solved with a maximum cost flow whose depth is the number of days:
	\begin{itemize}
		\item The days become a path, with capacity $u_i$,
		\item For each job, create a node (connected to source with capacity $1$) with two edges, one with cost $0$ to $t_i$, and 
		one with cost $c_i$ to $s_i$.
		\item Each node on the path has an edge to sink, $t$, with capacity equalling the number of jobs ending that day.
	\end{itemize}
	This leads to a network whose depth is the number of days,
	or number of stops.
	Both of these quantities are small in real instances.
	
	\item Balancing of schools (Application 15 from~\cite{AhujaMO95}): the goal is to assign students from varying
	backgrounds to schools, belong to several districts, so that
	each school as well as each district has students with balanced backgrounds.
	Furthermore, the cost corresponds to the distance the student needs to travel to the schools.
	This leads to a depth $3$ network where the layers are the students, schools, and districts respectively.
\end{itemize}


Finally, we believe several applications can benefit from the incorporation of non-linear weight functions.
As an instance, we can solve the following problem using our framework: Let $G$ be a network 
with $k$ sources, $\{s_i, 1 \leq i \leq k\}$, and a sink, $t$, and let $f_i$ be a concave function for each $1 \leq i \leq k$. We wish to maximize $\sum_i {f_i (F_i)}$, 
where $F_i$ is the amount of flow sent from $s_i$ to $t$.

\suppress{
\subsection{Optimal Closure and High Density Subgraphs}

Another motivating factor in this study is the recent
use of high density subgraphs in studying large
networks~\cite{GionisT15,MitzenmacherPPTX15,BhattacharyaHNT15}.

Here a common metric is max density subgraph, which
can be solved using a single parametric maximum flow,
or binary searching over maximum flows~\cite{Goldberg84}.
For simplicity and its relation to optimal closure,
we discuss the latter method.

One method is to create a $2$-layer graph where all edges
are on the first layer, and all vertices on the second.
The graph has three types of edges:
\begin{enumerate}
	\item $\alpha$ from source to $e$
	\item infinite from $e$ to its two vertices $u$ and $v$.
	\item $1$ from each $u$ to sink.
\end{enumerate}
A subgraph with density $\geq \alpha$ exists
iff the amount of flow is more than the capacity of things \vijay{things?}
leaving $S$.

This network construction is in fact a version of
optimum closure: here each node of the graph has
a potentially negative cost, and if a node is taken,
everything that goes into it must also be chosen.
This is also solved by connecting all negative nodes
to source, positives to sink, and running maxflow.

We believe the addition of costs gives a variety of
new ways of studying such networks.
However, closure problems are made difficult by the
fact that the answer depends on the gap between
the amount of flow into some edges and their capacities.
What we are able to prove is:
}

	\section{Negative Weight Functions}
\label{sec:negative}
We consider the case where our network can have negative weights. For simplicity, we only study linear weight functions. The analysis, however, can be extended to the case of concave function as in Section~\ref{sec:nonlinear}. Here, we assume that the absolute values of the weights are bounded:
\[ \ww_{\max} \geq |\ww_e| \geq \ww_{\min} \quad \forall e.\]
\begin{lemma} \label{lem:negative}
 Let $\xx$ be a solution returned by $\textsc{ScalingFlow}$, and let $\xx^*$ be the optimal solution
\[ \sum_e \ww_e \xx_e \geq (1- 8\epsilon) \sum_{\ww_e > 0} \ww_e \xx^*_e + (1 + 8\epsilon) \sum_{\ww_e < 0} \ww_e \xx^*_e.  \]
\end{lemma}
\begin{proof} 
A key observation is that the Lemmas \ref{lem:eligible} and \ref{lem:charge} are robust under negative weights. Using an argument similar to as in Theorem \ref{thm:factorLinear}, we have:
\begin{enumerate}
		\item If $\xx_e < \xx^*_e$, then $\xx_e < \cc_e$, and
		the invariant gives $\ww_e^{\pp} \leq 2 \delta_{T}$, which implies
		\[
		\ww_e^{\pp}\xx_e  \geq \left( \ww_e^{\pp} - 2 \delta_{T} \right) \xx_e  \geq  \left( \ww_e^{\pp} - 2 \delta_{T} \right) \xx^*_e.
		\]
		Since $\delta(T) \leq \epsilon \ww_{\min} \leq \epsilon |\ww_e|$ for all $e$, we consider 2 cases: 
		\[
		\ww_e^{\pp}\xx_e  \geq \left( \ww_e^{\pp} - 2 \delta_{T} \right) \xx^*_e \geq 
		\begin{cases}
		\ww_e^{\pp}\xx^*_e - 2\epsilon \ww_e \xx^*_e \quad \text{if $\ww_e > 0$,}\\
		\ww_e^{\pp}\xx^*_e + 2\epsilon \ww_e \xx^*_e  \quad \text{if $\ww_e < 0$.}
		\end{cases}
		\]
		\item If $\xx_e > \xx^*_e$, then $\xx_e > 0$, and the invariant gives $\ww_e^{\pp} \geq -3 l_e \delta_e^f$, which combined with $\xx^*_e \leq \xx_e$ gives:
		\[
		\xx_e \left( \ww_e^{\pp} + 3 l_e \delta_e^f \right) \geq \xx^*_e \left( \ww_e^{\pp} + 3 l_e \delta_e^f \right) \geq \xx^*_e \ww_e^{\pp}
		\]
\end{enumerate}
Combining these then gives
\begin{align*}
	\sum_{e} \xx_e \left( \ww_e + 3 l_e\delta_e^f  \right) & = \sum_{e} \xx_e \left( \ww_e^{\pp} + 3 l_e\delta_e^f  	\right)  \\
	& \geq \sum_{e}  \ww_e^{\pp}\xx^*_e - \sum_{e:\ww_e < 0} 2\epsilon\ww_e \xx^*_e + \sum_{e:\ww_e < 0}  2\epsilon \ww_e \xx^*_e	\\
	& \geq \sum_{e:\ww_e > 0} \left( \ww_e\xx^*_e - 2\epsilon \ww_e \xx^*_e \right)+ \sum_{e:\ww_e < 0} \left( \ww_e\xx^*_e + 2\epsilon \ww_e \xx^*_e	 \right)\\
	& = (1 - 2 \epsilon) \sum_{e:\ww_e > 0} \ww_e  \xx^*_e + (1 + 2 \epsilon) \sum_{e:\ww_e < 0} \ww_e  \xx^*_e
\end{align*}
and the result then follows from Lemma~\ref{lem:charge}.

\end{proof}

\begin{proof}(Of Lemma~\ref{lem:mincost}) Let $\ww_e = - \qq_e$ for every $e$, and add a directed edge from $t$ to a dummy vertex $t'$ with $\ww_{tt'} = Q$. By Lemma~\ref{lem:negative}, a solution returned by $\textsc{ScalingFlow}$ on that instance has the following guarantee:
\[Qf - \sum_e \qq_e \xx_e \geq (1 - \epsilon) Qf^* - (1 + \epsilon) \sum \qq_e \xx^*_e. \]
Negating both sides gives
\begin{align*}
\sum_e \qq_e \xx_e - Qf & \leq (1 + \epsilon) \sum \qq_e \xx^*_e - (1 - \epsilon) Qf^* \\
& = (1 + \epsilon) \left( \sum \qq_e \xx^*_e - Qf^* \right) + 2 \epsilon Qf^*.
\end{align*}

\end{proof}
	
\end{document}